\documentclass[journal]{IEEEtran}

\usepackage[numbers,sort&compress,square]{natbib}   
\usepackage{graphicx}
\usepackage{subfigure}
\usepackage{amsmath,amssymb}
\usepackage{amsfonts}
\usepackage{array}
\usepackage{amsthm}
\usepackage{mathrsfs}
\usepackage{setspace}
\usepackage{color}
\usepackage[ruled,linesnumbered]{algorithm2e}
\usepackage{soul}

\newtheorem{theorem}{Theorem}
\newtheorem{proposition}[theorem]{Proposition}
\newtheorem{lemma}[theorem]{Lemma}

\theoremstyle{definition}
\newtheorem{define}[theorem]{Definition}

\newtheorem{remark}[theorem]{Remark}
\newtheorem{example}[theorem]{Example}

\begin{document}

	\title{Scheduling  Stochastic Real-Time Jobs \\in Unreliable Workers }

%

\author{\IEEEauthorblockN{Yu-Pin Hsu, Yu-Chih Huang, and Shin-Lin Shieh}
	
	\thanks{Y.-P. Hsu, Y.-C. Huang, and S.-L. Shieh are with Department of Communication Engineering, National Taipei University, Taiwan. Email: \texttt{\{yupinhsu, ychuang, slshieh\}@mail.ntpu.edu.tw}. The work was supported by Ministry of Science and Technology of Taiwan under Grant MOST 107-2221-E-305-007-MY3.} 
	
}

\maketitle

\begin{abstract}
We consider a distributed computing network consisting of a master  and multiple workers  processing tasks of \textit{different types}. The master  is running multiple applications. Each application  \textit{stochastically} generates  \textit{real-time jobs}  with a strict job deadline, where each job is a collection of tasks of some types specified by the application. A real-time job is completed only when \textit{all} its tasks are completed by the corresponding workers within the deadline. Moreover, we consider \textit{unreliable} workers, whose processing speeds are uncertain. Because of the limited processing abilities of the workers, an algorithm for scheduling the jobs in the  workers is needed to maximize the average number of completed jobs for each application. The scheduling problem is not only critical but also  practical in  distributed computing networks. In this paper, we develop two scheduling algorithms, namely, a \textit{feasibility-optimal scheduling algorithm} and an \textit{approximate scheduling algorithm}. The feasibility-optimal scheduling algorithm can fulfill the largest region of  applications' requirements for the average number of completed jobs. However, the feasibility-optimal scheduling algorithm  suffers from high computational complexity when the number of applications is large. To address the issue, the approximate scheduling algorithm is  proposed with a guaranteed approximation ratio in the worst-case scenario. The approximate scheduling algorithm is also validated in the average-case scenario via computer simulations.  
\end{abstract}


\section{Introduction} \label{section:introduction}
\begin{figure}[h]
	\centering
	\includegraphics[width=.45\textwidth]{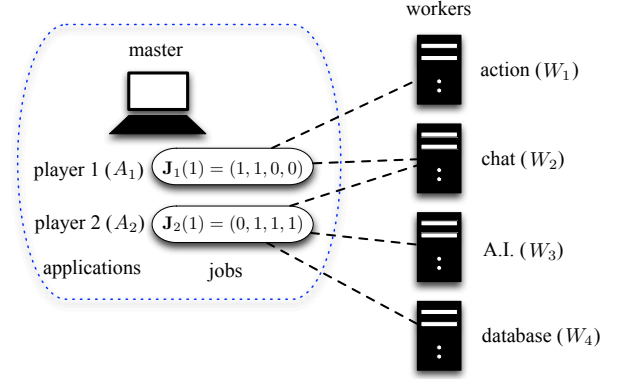}
	\caption{A distributed computing network with two ($N=2$) applications $A_1$ and $A_2$ and four ($M=4$) workers $W_1$, $W_2$,  $W_3$, and $W_4$. At the beginning of frame 1, application $A_1$ generates job $\mathbf{J}_1(1)$ with tasks for workers $W_1$ and $W_2$, and application $A_2$ generates  job $\mathbf{J}_2(1)$ with tasks for workers $W_2$, $W_3$, and $W_4$.}
	\label{fig:network}
\end{figure}
Distributed computing networks (such as MapReduce \cite{dean2008mapreduce}) become increasingly popular to support data-intensive  jobs. The underlying idea to process a data-intensive job is to  divide  the job  into a group of small tasks that can be  processed in parallel by multiple workers. In general, a worker can be specialized to process a type of tasks. For example, MapReduce allows an application to specify its computing network. Another outstanding example is distributed computing networks for massive multiplayer online games \cite{matskin2003scalable}. The online game system illustrated in Fig.~\ref{fig:network} includes one master and four workers processing different types of tasks. The master is serving two players.  While the present job of player~1 needs two types of workers to get completed, that of player 2 needs three types of workers.

Moreover, because of the real-time nature of latency-intensive applications (e.g., online games), a real-time job needs to be completed in a deadline. To maximize the number of  jobs that meet the deadline,  a  scheduling algorithm allocating  workers to jobs is needed. Job-level scheduling poses more challenges than packet-level scheduling. That is because all  tasks in a job are \textit{dependent} in the sense that a job is not completed  until \textit{all} its tasks are completed, but all packets or tasks in traditional packet-based  networks are independently treated.

Most prior research on job-level scheduling considered general-purpose workers. The closest scenario to ours (i.e., specialized workers) is the \textit{coflow} model proposed in \cite{chowdhury2012coflow}, where a coflow is a job consisting of tasks of various types. Since the coflow model was proposed, coflow scheduling has been a hot topic, e.g. \cite{chowdhury2014efficient,li2016efficient,shafiee2018improved,tseng2018coflow,im2019matroid}. See the recent survey paper \cite{wang2018survey}. However, almost all prior research on the coflow scheduling focused on \textit{deterministic} networks; in contrast, little attention was given to  \textit{stochastic} networks.  Note that a job can be randomly generated; moreover, a worker can be unreliable because of unpredictable events \cite{ananthanarayanan2013effective} like hardware failures.  Because of the practical issues,  a scheduling algorithm for \textit{stochastic} real-time jobs in \textit{unreliable} workers is crucial in distributed computing networks. The most relevant works to ours  are \cite{liang2017coflow,li2018efficient}. While \cite{liang2017coflow} focused on homogeneous stochastic jobs in the coflow model, \cite{li2018efficient} extended to a heterogeneous case. The fundamental difference between those relevant works  and ours is that we consider stochastic \textit{real-time} jobs and \textit{unreliable} workers.

In this paper, we consider a master and $M$ specialized workers. The master is running multiple applications, which  stochastically generate real-time jobs with a hard deadline. The workers are unreliable. Our main contribution lies in developing job scheduling algorithms  with provable performance guarantees. Leveraging Lyapunov techniques, we propose a feasibility-optimal scheduling algorithm for maximizing the region of achievable requirements for the average number of completed jobs. However, the feasibility-optimal scheduling algorithm turns out to involve an NP-hard combinatorial optimization problem. To tackle the computational issue, we propose an approximate scheduling algorithm that is computationally tractable; furthermore, we prove that its region of achievable requirements  shrinks by   a factor of at most $1/\sqrt{M}$ from the largest one. More  surprisingly, our simulation results show that the region of achievable requirements  by the approximate scheduling algorithm is close to the largest one.

\section{System overview} \label{section:system}

\subsection{Network model} \label{subsection:network}

Consider a distributed computing network consisting of a master  and $M$ specialized workers $W_1, \cdots, W_M$. The master  is running~$N$ applications $A_1, \cdots, A_N$. Fig.~\ref{fig:network} illustrates an example network with $N=2$ and $M=4$. Suppose that data transfer between the master  and the workers occurs instantaneously with no error. Note that the prior works on the coflow model focused on the time for data transfer. To investigate the unreliability of the workers, we ignore the time for data transfer; instead, focus on the time for computation. 

Divide time into frames and index them by $t=1, 2, \cdots$. At the beginning of each frame~$t$, each application   \textit{stochastically} generates a  job, where a job is a collection of tasks that can be processed by the corresponding workers. Precisely, we use  vector $\mathbf{J}_{i}(t)=(J_{i,1}(t), \cdots, J_{i,M}(t))$ to represent the job generated by application $A_i$ in frame~$t$, where each element $J_{i,j}(t)\in \{0, 1\}$ indicates if the job  has a task for worker $W_j$: if $J_{i,j}(t)=1$, then the job $\mathbf{J}_i(t)$ has a task for  worker $W_j$; otherwise, it does not. See Fig.~\ref{fig:network} for example. Each task is also \textit{stochastically} generated, i.e.,  $J_{i,j}(t)$ is a random variable for all $i$, $j$, and~$t$. By $|\mathbf{J}_i(t)|$ we denote the number of 1's in  vector $\mathbf{J}_i(t)$; in particular, if $|\mathbf{J}_i(t)|=0$, then application $A_i$ generates no job in frame~$t$. Suppose that the probability distribution of random variable $J_{i,j}(t)$ is independently and identically
distributed (i.i.d.) over frame~$t$, for all $i$ and $j$. Suppose that the tasks generated by application $A_i$ for worker $W_j$ have the same workload. See Remark~\ref{remark:time-varying-work} later for time-varying workloads. Moreover, the jobs need real-time computations. Suppose that the deadline for each  job is one frame. The real-time system has been justified in the literature, e.g., see \cite{hou2013packets}.

%

Consider a time-varying processing speed for each worker. Suppose that  the processing speed of each worker is  i.i.d. over frames. 
With the i.i.d. assumption along with those constant workloads, we can assume that a task generated  by application $A_i$ can be completed by  worker $W_j$ (i.e., when $J_{i,j}(t)=1$) with a constant probability $P_{i,j}$ over frames.  At the end of each frame, each  worker reports if its task is completed in that frame. A job is completed only when all its tasks are completed in the arriving frame.  If any task of a  job cannot be completed  in the arriving frame, the job expires and is removed from the application.

Unaware of the completion of a task at the beginning of each frame, we suppose that the master   assigns at most one task to a worker for each frame. If two jobs $\mathbf{J}_i(t)$ and $\mathbf{J}_j(t)$, for some $i$ and $j$, need the same worker in frame $t$, i.e., $J_{i,k}(t)=J_{j,k}(t)=1$  for some $k$, then we say the two jobs have \textit{interference}. For example, jobs $\mathbf{J}_1(1)$ and $\mathbf{J}_2(1)$ in Fig.~\ref{fig:network} have the interference. 

As a result of the interference, the master  has to decide a set  of interference-free jobs for computing in each frame. Let $\mathbf{D}(t) \subseteq \{\mathbf{J}_1(t), \cdots, \mathbf{J}_N(t)\}$ be the set of interference-free jobs decided for computing in frame~$t$. For example, decision $\mathbf{D}(1)$ in Fig.~\ref{fig:network} can be either $\mathbf{J}_1(1)$ or $\mathbf{J}_2(1)$. If $\mathbf{D}(1)=\mathbf{J}_1(1)$ in Fig.~\ref{fig:network}, then workers $W_1$ and $W_2$ are allocated to job $\mathbf{J}_1(1)$ in frame~$1$; moreover, job $\mathbf{J}_1(1)$ is completed only when the two workers complete their respective tasks in frame~1. A scheduling algorithm $\pi=\{\mathbf{D}(1), \mathbf{D}(2), \cdots\}$ is a time sequence of the decisions for all frames. 
%

\subsection{Problem formulation} \label{subsection:problem}
Let random variable $e_i(t;\pi) \in \{0, 1\}$ indicate if job $\mathbf{J}_i(t)$ is completed in frame~$t$ under scheduling algorithm~$\pi$, where  $e_i(t;\pi)=1$ if job $\mathbf{J}_i(t)$ is generated (i.e., $|\mathbf{J}_i(t)|\neq 0$) and all  tasks of the job are completed by the corresponding workers in frame~$t$; $e_i(t;\pi)=0$ otherwise. The random variable $e_i(t;\pi)$ depends on the random variables $J_{i,j}(t)$, the task completion probabilities $P_{i,j}$, and a potential randomized scheduling algorithm~$\pi$. 

We define the average number $\mathscr{N}_i(\pi)$ of completed jobs for application $A_i$ under scheduling algorithm $\pi$ by
\begin{align}
\mathscr{N}_i(\pi)=\liminf_{T \rightarrow \infty} \frac{\sum_{t=1}^T E[e_i(t;\pi)]}{T}.  \label{eq:completion-rate}
\end{align}
Let vector~$\mathbf{r}=(r_1, \cdots, r_N)$  represent an applications' \textit{requirement} for the average numbers of completed jobs. We say that requirement~$\mathbf{r}$ can be \textit{fulfilled} (or \textit{achieved}) by scheduling algorithm~$\pi$ if $\mathscr{N}_i(\pi) \geq r_i$ for all~$i$. Moreover, we refer to requirement~$\mathbf{r}$ as a \textit{feasible requirement} if there exists \textit{a} scheduling algorithm that can fulfill the requirement. 
We define the \textit{maximum feasibility region} as follows. 
\begin{define}
The \textbf{maximum feasibility region} $\mathbf{R}_{\max}$ is the ($N$-dimensional) region consisting of all feasible requirements~$\mathbf{r}$. 
\end{define} 

We define an optimal scheduling algorithm as follows.
\begin{define}
A scheduling algorithm~$\pi$ is called a \textbf{feasibility-optimal}\footnote{The feasibility-optimal scheduling defined in this paper  is analogous to the throughput-optimal scheduling (e.g., \cite{neely2010stochastic}) or the timely-throughput-optimal scheduling (e.g., \cite{hou2013packets}).}  scheduling algorithm if, for any requirement~$\mathbf{r}$ \textit{interior}\footnote{We say that requirement $\mathbf{r}=(r_1, \cdots, r_N)$ is \textit{interior} of the region $\mathbf{R}_{\max}$ if there exists an $\epsilon>0$ such that $\mathbf{r}+\epsilon=(r_1+\epsilon, \cdots, r_N+\epsilon)$ lies in the region $\mathbf{R}_{\max}$. The concept of the \textit{strictly} feasible requirement has been widely used in the throughput-optimal scheduling or timely-throughput-optimal scheduling.} of $\mathbf{R}_{\max}$, it can be fulfilled by the scheduling algorithm~$\pi$. 
\end{define}

The goal of this paper is to devise a \textit{feasibility-optimal} scheduling algorithm.

%

\section{Scheduling algorithm design} \label{section:design}

In this section, we develop a feasibility-optimal scheduling algorithm for managing the stochastic real-time jobs in the unreliable workers. To that end, we introduce a virtual queueing network in Section~\ref{subsection:virtual-queue}. With the assistance of the virtual queueing network, we propose a feasibility-optimal  scheduling design in Section~\ref{subsection:optimal-design}. However, the proposed feasibility-optimal scheduling algorithm involves a combinatorial optimization problem.  We show that the combinatorial optimization problem is NP-hard. Thus, we develop a tractable \textit{approximate  scheduling algorithm} in Section~\ref{subsection:approximate}; meanwhile, we establish its approximation ratio.

\subsection{Virtual queueing network} \label{subsection:virtual-queue}

Given the distributed computing network  with  scheduling algorithm~$\pi$ and requirement~$\mathbf{r}$, we construct a virtual queueing network. The virtual queueing network  consists of~$N$ queues $Q_1, \cdots, Q_N$, operating under the same frame system as that in Section~\ref{subsection:network}. For example, Fig.~\ref{fig:queue} is the virtual queueing network for the distributed computing network in Fig.~\ref{fig:network}.  We want to emphasize that the virtual queueing network is not a real-world network. It is introduced for the scheduling design in Section~\ref{subsection:optimal-design}.

\begin{figure}[t]
	\centering
	\includegraphics[width=.35\textwidth]{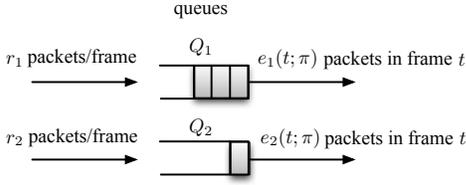}
	\caption{Virtual queueing network for the distributed computing network in Fig.~\ref{fig:network} with scheduling algorithm~$\pi$ and requirement $(r_1, r_2)$.}
	\label{fig:queue}
\end{figure}

At the beginning of each frame~$t$, a fixed number\footnote{The virtual queueing network has a fractional number of packets.} $r_i$ of packets arrive at  queue $Q_i$. 
  At the end of frame~$t$, queue $Q_i$ can remove $e_i(t;\pi)$ packet, i.e., if job $\mathbf{J}_i(t)$ is completed in frame~$t$, then queue $Q_i$ can remove one packet at the end of frame~$t$; otherwise, it  removes no packet in frame~$t$. Again, note that those packets are not real-world packets.
  We summarize the packet arrival rate and the packet service rate as follows.

 \begin{proposition}\label{proposition:arrival-service}
 The packet arrival rate for queue~$Q_i$ is~$r_i$, and the  packet service rate for queue~$Q_i$ is $\mathscr{N}_i(\pi)$. 	
 \end{proposition}
\begin{proof}
The packet arrival rate  for $Q_i$ is $\lim_{T \rightarrow \infty}\frac{\sum_{i=1}^T r_i}{T}=r_i$. The packet service rate for $Q_i$ is $\liminf_{T \rightarrow \infty}\frac{\sum_{t=1}^T E[e_i(t;\pi)]}{T}=\mathscr{N}_i(\pi)$.	
\end{proof}
 
Let $Q_i(t)$ be the queue size at queue $Q_i$ at the beginning (before new packet arrival) of frame~$t$. Then, the queueing dynamics of queue $Q_i$ can be expressed by $Q_i(t+1)=\max\{Q_i(t)+r_i-e_i(t;\pi),0\}$. Let $\mathbf{Q}(t)=(Q_1(t), \cdots, Q_N(t))$ be the vector of all queue sizes at the beginning of frame~$t$. We define the notion of a \textit{stable}  queue in Definition~\ref{define:stable}, followed by introducing a necessary condition for the stable queue in Proposition~\ref{proposition:stable}.
  
  \begin{define} \label{define:stable}
  Queue $Q_i$ is \textbf{stable} if the average queue size $\limsup_{T \rightarrow \infty}  \frac{\sum_{t=1}^T E[Q_i(t)]}{T}$ is finite.
  \end{define}

  \begin{proposition}[\cite{georgiadis2006resource}, Lemma 3.6]\label{proposition:stable}
  If queue $Q_i$ is stable, then its packet service rate is greater than or equal to its packet arrival rate.  	
  \end{proposition}

 By Propositions~\ref{proposition:arrival-service} and~\ref{proposition:stable}, we can turn our  attention to developing a scheduling algorithm  such that, for any requirement~$\mathbf{r}$ interior of $\mathbf{R}_{\max}$, all  queues in the virtual queueing network are stable.

We want to emphasize that, unlike traditional  stochastic  networks (e.g., \cite{neely2010stochastic,hou2013packets}),   each packet in our virtual queueing network can be removed only when \textit{all}  associated tasks are completed in its arriving frame.  Thus, our paper generalizes to stochastic  networks with \textit{multiple required servers}; in particular, we develop a tractable approximate scheduling algorithm for the scenario in Section~\ref{subsection:approximate}. 

\subsection{Feasibility-optimal scheduling algorithm} \label{subsection:optimal-design}

\begin{algorithm}[t]
	\SetAlgoLined 
	\SetKwFunction{Union}{Union}\SetKwFunction{FindCompress}{FindCompress} \SetKwInOut{Input}{input}\SetKwInOut{Output}{output}
%
	\tcc{At the beginning of  frame $1$, perform as follows:}
	$Q_i(1) \leftarrow 0$\,\,for all $i$\;\label{alg:max-weight-q-int}
	
	\tcc{At the beginning of each frame $t=1, 2, \cdots$, perform as follows:}
	
		$Q_i(t) \leftarrow Q_i(t-1)+r_i$\,\, for all $i$\; \label{alg:max-weight:q1}
	
	Perform a set $\mathbf{D}(t) \subseteq \{\mathbf{J}_1(t), \cdots, \mathbf{J}_N(t)\}$ of interference-free jobs such that 
	\begin{align}
	\sum_{i: \mathbf{J}_i(t) \in \mathbf{D}(t)} Q_i(t)  \left(\mathbf{1}_{|\mathbf{J}_i(t)| \neq 0}\cdot   \prod_{j:J_{i,j}(t)=1}P_{i,j}\right)  \label{eq:opt}
	\end{align}
	is maximized\;  \label{alg:max-weight:d}
	
		\tcc{At the end of each frame $t$, perform as follows:}
	\For{$i=1$ \KwTo $N$\label{alg:max-weight:q2-start}}{
	\If{$\mathbf{J}_{i}(t) \in \mathbf{D}(t)$, $|\mathbf{J}_{i}(t)| \neq 0$, and \text{all its workers complete their respective tasks}}{
		$Q_i(t) \leftarrow \max\{Q_i(t)-1,0\}$\; \label{alg:max-weight:q2}	
	}	
	\label{alg:max-weight:q2-end}}
	\caption{Feasibility-optimal scheduling algorithm.}
	\label{alg:max-weight}
\end{algorithm}

In this section, we propose a feasibility-optimal scheduling algorithm in Alg.~\ref{alg:max-weight}. At the beginning of frame~1, Alg.~\ref{alg:max-weight} (in Line~\ref{alg:max-weight-q-int}) initializes all  queue sizes to be zeros. At the beginning of each frame~$t$, Alg.~\ref{alg:max-weight}  (in Line~\ref{alg:max-weight:q1}) updates each queue $Q_i$  with the new arriving~$r_i$ packets; then, Alg.~\ref{alg:max-weight} (in Line~\ref{alg:max-weight:d}) decides $\mathbf{D}(t)$  for that frame according to the \textit{present} queue size vector~$\mathbf{Q}(t)$. The decision $\mathbf{D}(t)$ is made for maximizing the weighted sum of the queue sizes in Eq.~(\ref{eq:opt}). The term $\mathbf{1}_{|\mathbf{J}_i(t)|\neq 0} \cdot \prod_{j:J_{i,j}(t)=1}P_{i,j}$ in Eq.~(\ref{eq:opt}) calculates the expected packet service rate for $Q_i$,  where the indicator function $\mathbf{1}_{|\mathbf{J}_i(t)|\neq 0}$ indicates if job $\mathbf{J}_i(t)$ is generated in frame~$t$, and if so, that job can be completed with probability  $\prod_{j:J_{i,j}(t)=1}P_{i,j}$. The underlying idea of Alg.~\ref{alg:max-weight} is to remove as many packets from the virtual queueing network as possible (for stabilizing all queues).

After performing the decision $\mathbf{D}(t)$, Alg.~\ref{alg:max-weight} (in Line~\ref{alg:max-weight:q2}) updates each $Q_i$  at the end of frame~$t$: if job $\mathbf{J}_i(t)$ is scheduled, the job is indeed generated, and all its required workers complete their respective tasks, then one packet is removed from queue~$Q_i$ in the virtual queuing network. 

\begin{example}\label{ex:1}
Take Figs.~\ref{fig:network} and \ref{fig:queue} for example. Suppose that $P_{1,j}=0.8$ and $P_{2,j}=0.9$ for all $j$, and $\mathbf{r}=(0.48, 0.5)$. According to Line~\ref{alg:max-weight:d}, Alg.~\ref{alg:max-weight} calculates $Q_1(1)\prod_{j:J_{1,j}(1)=1}P_{1,j}=0.48\cdot0.8^2=0.3072$ and $Q_2(1)\prod_{j:J_{2,j}(1)=1}P_{2,j}=0.5 \cdot 0.9^3=0.3645$. Thus, Alg.~\ref{alg:max-weight} decides to compute $\mathbf{J}_2(1)$ for frame~1. If workers $W_2$, $W_3$, and $W_4$ in Fig.~\ref{fig:network} can complete their respective tasks in frame~1, then  one packet is removed from queue $Q_2$ in Fig.~\ref{fig:queue} at the end of frame~1, i.e., queue $Q_2$ has $\max\{0.5-1,0\}=0$ packet at the end of frame~1.  
\end{example}

Leveraging Lyapunov techniques \cite{neely2010stochastic}, we can establish the optimality of Alg.~\ref{alg:max-weight} in the following. 

\begin{theorem} \label{theorem:optimal}
Alg.~\ref{alg:max-weight} is a feasibility-optimal scheduling algorithm.  
\end{theorem}      
\begin{proof}
Let vector $(\mathbf{J}_1(t), \cdots, \mathbf{J}_N(t))$ represent the \textit{state} of the virtual queueing network in frame~$t$. Note that the state changes over frames but its probability distribution is i.i.d., according to the assumption in Section~\ref{subsection:network}. Following the  standard argument of the Lyapunov theory in \cite[Chapter 4]{neely2010stochastic} along with  the i.i.d. property of the state, we can prove that for any requirement~$\mathbf{r}$ interior of $\mathbf{R}_{\max}$, all queues in the virtual queueing network (associated with Alg.~\ref{alg:max-weight}) are stable. That is, Alg.~\ref{alg:max-weight} can fulfill the requirement~$\mathbf{r}$ by Propositions~\ref{proposition:arrival-service} and~\ref{proposition:stable}. Thus, Alg.~\ref{alg:max-weight} is feasibility-optimal.
\end{proof}              

Note that Alg.~\ref{alg:max-weight} involves a combinatorial optimization problem in Line~\ref{alg:max-weight:d}. In the next section, we will investigate the computational complexity for solving the combinatorial optimization problem.  

\subsection{Tractable approximate scheduling algorithm} \label{subsection:approximate}
We show (in the next lemma) that the combinatorial optimization problem in Line~\ref{alg:max-weight:d} of Alg.~\ref{alg:max-weight} is NP-hard. Therefore, Alg.~\ref{alg:max-weight} is computationally intractable. 

\begin{lemma}\label{lemma:np}
The combinatorial optimization problem in Alg.~\ref{alg:max-weight} in frame~$t$ is NP-hard, for all $t$. 	
\end{lemma}
\begin{proof}
We construct a reduction from the set packing problem \cite{halldorsson1998independent}.
See Appendix~\ref{appendix:lemma:np} for details.
\end{proof}

\begin{algorithm}[t]
	\SetAlgoLined 
	\SetKwFunction{Union}{Union}\SetKwFunction{FindCompress}{FindCompress} \SetKwInOut{Input}{input}\SetKwInOut{Output}{output}
	%

	\tcc{At the beginning of  frame $1$, perform as follows:}
$Q_i(1) \leftarrow 0$\,\,for all $i$\;	
	
	\tcc{At the beginning of each frame $t=1, 2, \cdots$, perform as follows:}

	$Q_i(t) \leftarrow Q_i(t-1)+r_i$\,\,for all $i$\;
		
	$\mathbf{U} \leftarrow \{W_1, \cdots, W_M\}$\; \label{alg:approximate:u-initial}
	$\mathbf{D}(t) \leftarrow \emptyset$\;

	Sort  all  jobs $\mathbf{J}_i(t)$ according to the values of 
	\begin{align}
	\left\{
	\begin{array}{ll}
	\frac{Q_i(t) \prod_{j:J_{i,j}(t)=1}P_{i,j}}{\sqrt{|\mathbf{J}_i(t)|}} & \text{if $|\mathbf{J}_i(t)| \neq 0$;}\\
	0 & \text{else,}
	\end{array}
	\right.
	\label{eq:apx}
	\end{align}
	 to obtain the sorted jobs $\mathbf{J}^{(1)}(t), \cdots, \mathbf{J}^{(N)}(t)$\;  \label{alg:approximate:value}
	
	\For{$i=1$ \KwTo $N$ \label{alg:approximate:for-begin}}{
		\If{$|\mathbf{J}^{(i)}(t)|\neq 0$  and $\{W_j: J^{(i)}_{j}(t)=1\} \subseteq \mathbf{U}$\label{alg:approximate:condition}}{
			$\mathbf{D}(t) \leftarrow \mathbf{D}(t) \cup \mathbf{J}^{(i)}(t)$\; \label{alg:approximate:d}
			$\mathbf{U} \leftarrow \mathbf{U}-\{W_j: J^{(i)}_{j}(t)=1\}$\; \label{alg:approximate:u}
			
		}
	
	\label{alg:approximate:for-end}}

	Perform the decision $\mathbf{D}(t)$\; \label{alg:approximate:perform}
		\tcc{At the end of each frame $t$, perform as follows:}
\For{$i=1$ \KwTo $N$}{
	\If{$\mathbf{J}_i(t) \in \mathbf{D}(t)$, $|\mathbf{J}_{i}(t)| \neq 0$, and \text{all its workers complete their respective tasks}}{
		$Q_i(t) \leftarrow \max\{Q_i(t)-1,0\}$\; 	\label{alg:approximate:q}
	}	
	}
	\caption{Approximate scheduling algorithm.}
	\label{alg:approximate}
\end{algorithm}

To study the NP-hard problem, we define two notions of  approximation ratios as follows. While Definition~\ref{define:apx1} studies the resulting value in Eq.~(\ref{eq:opt}), Definition~\ref{define:apx2} investigates the resulting region of achievable requirements. 

\begin{define}\label{define:apx1}
Given queue size vector $\mathbf{Q}(t)$ in frame~$t$. Let $OPT(t)$ be the value in Eq.~(\ref{eq:opt}) computed by Alg.~\ref{alg:max-weight} in frame~$t$. Let $APX(t;\pi)$ be the value in Eq.~(\ref{eq:opt}) computed by scheduling algorithm~$\pi$ in frame $t$. Then, the scheduling algorithm~$\pi$ is called a \textbf{$p$-approximate scheduling algorithm to Eq.~(\ref{eq:opt})} if $OPT(t)/APX(t;\pi) \leq p$ for  all possible $\mathbf{Q}(t)$ and $t$.
\end{define}

\begin{define}\label{define:apx2}
A scheduling algorithm~$\pi$ is called a \textbf{$p$-approximate scheduling algorithm to $\mathbf{R}_{\max}$} if, for any requirement $\mathbf{r}=(r_1, \cdots, r_N)$ interior of $\mathbf{R}_{\max}$,  requirement $\mathbf{r}/p=(r_1/p, \cdots, r_N/p)$  can be fulfilled by the scheduling algorithm~$\pi$. 
\end{define}

In this paper, we propose an approximate scheduling algorithm in Alg.~\ref{alg:approximate}. The procedure of Alg.~\ref{alg:approximate} is similar to that of Alg.~\ref{alg:max-weight}; hence, we point out key differences in the following. 

Unlike Alg.~\ref{alg:max-weight} solving the combinatorial optimization problem, Alg.~\ref{alg:approximate} (in Line~\ref{alg:approximate:value}) simply sorts all jobs according to the values computed by Eq.~(\ref{eq:apx}).  Let $\mathbf{J}^{(1)}(t), \cdots, \mathbf{J}^{(N)}(t)$ (in Line~\ref{alg:approximate:value}) denote the sorted jobs in frame~$t$ in  descending order of the values from Eq.~(\ref{eq:apx}). In addition, let $J^{(i)}_j(t)$    (in Line~\ref{alg:approximate:condition}) indicate if job $\mathbf{J}^{(i)}$ has a task for worker $W_j$ in frame~$t$. While the numerator of Eq.~(\ref{eq:apx}) indicates the weight $Q_i(t)\prod_{j:J_{i,j}(t)=1}P_{i,j}$ in Eq.~(\ref{eq:opt}) for job $\mathbf{J}_i(t)$,  the denominator of that reflects the maximum number of jobs interfered by job $\mathbf{J}_i(t)$.
The underlying idea of Alg.~\ref{alg:approximate} is to consider jobs in $\mathbf{J}^{(1)}(t), \cdots, \mathbf{J}^{(N)}(t)$ order, for achieving a higher value of  Eq.~(\ref{eq:opt}) and at the same time keeping the interference as low as possible.

More precisely, Alg.~\ref{alg:approximate} uses a set $\mathbf{U}$ to record (in Line~\ref{alg:approximate:u}) the \textit{available} workers that are not allocated yet, where set $U$ is initialized to be $\{W_1, \cdots, W_M\}$  in Line~\ref{alg:approximate:u-initial}. Then, at the $i$-th iteration of Line~\ref{alg:approximate:for-begin}, Alg.~\ref{alg:approximate} checks if  job $\mathbf{J}^{(i)}(t)$  satisfies the two conditions in Line~\ref{alg:approximate:condition}: the first condition $|\mathbf{J}^{(i)}(t)|\neq 0$ means that job $\mathbf{J}_i(t)$ is  generated and the second condition $\{W_j: J^{(i)}_{j}(t)=1\} \subseteq \mathbf{U}$ means that  its required workers are all available. If  job $\mathbf{J}^{(i)}(t)$ meets the conditions, then it is scheduled as in Line~$\ref{alg:approximate:d}$. In addition, if job $\mathbf{J}^{(i)}(t)$ is scheduled, then set $\mathbf{U}$ is updated as in Line~\ref{alg:approximate:u} by removing the workers allocated to job $\mathbf{J}^{(i)}(t)$. After deciding $\mathbf{D}(t)$, Alg.~\ref{alg:approximate}  performs the decision $\mathbf{D}(t)$ in Line~\ref{alg:approximate:perform} for frame~$t$, followed by updating the queue sizes in Line~\ref{alg:approximate:q}.

\begin{example}
Follow Ex.~\ref{ex:1}. According to Eq.~(\ref{eq:apx}),  Alg.~\ref{alg:approximate} calculates $\frac{Q_1(1) \prod_{j:J_{1,j}(1)=1}P_{1,j}}{\sqrt{|\mathbf{J}_1(1)|}}=\frac{0.48\cdot 0.8^2}{\sqrt{2}}=0.2172$ and $\frac{Q_2(1) \prod_{j:J_{2,j}(1)=1}P_{2,j}}{\sqrt{|\mathbf{J}_2(1)|}}=\frac{0.5\cdot 0.9^3}{\sqrt{3}}=0.2104$. Thus, Alg.~\ref{alg:approximate} decides to compute $\mathbf{J}_1(1)$ for frame~1. Note that the decision is different from that in Ex.~\ref{ex:1}. 
\end{example}

Next, we establish the approximation ratio of Alg.~\ref{alg:approximate} to Eq.~(\ref{eq:opt}). 
\begin{lemma}\label{lemma:apx}
Alg.~\ref{alg:approximate} is a $\sqrt{M}$-approximate scheduling algorithm to Eq.~(\ref{eq:opt}). 
\end{lemma}
\begin{proof}
See Appendix~\ref{appendix:lemma:apx}.
\end{proof}
\begin{remark}
We remark that the approximation ratio of $\sqrt{M}$ is the best approximation ratio to Eq.~(\ref{eq:opt}). That is because the combinatorial optimization problem in  Alg.~\ref{alg:max-weight} is computationally harder than the set packing problem (see Lemma~\ref{lemma:np}) and the best approximation ratio to the set packing problem is the square root (see \cite{halldorsson1998independent}).
\end{remark}

With  Lemma~\ref{lemma:apx}, we can further establish the approximation ratio of Alg.~\ref{alg:approximate} to $\mathbf{R}_{\max}$. 
\begin{theorem}\label{theorem:approximation-ratio}
Alg.~\ref{alg:approximate} is a $\sqrt{M}$-approximate scheduling algorithm to $\mathbf{R}_{\max}$. 
\end{theorem}
\begin{proof}
See Appendix~\ref{appendix:theorem:approximation-ratio}.
\end{proof}

The computational complexity of Alg.~\ref{alg:approximate} is $O(\log N)$ primarily caused by sorting all queues in Line~\ref{alg:approximate:value}. Thus, Alg.~\ref{alg:approximate} is tractable when the number of applications is large.

\begin{remark}\label{remark:time-varying-work}
	We  remark that our methodology can apply to the case of time-varying workloads. Let $L_{i,j}(t)$ be the workload generated by application $A_i$ for worker $W_j$ in frame~$t$. We just need to revise the constant task completion probability $P_{i,j}$ in Algs.~\ref{alg:max-weight} and \ref{alg:approximate} to be the probability of completing workload $L_{i,j}(t)$. If workload $L_{i,j}(t)$ is i.i.d. over frames~$t$ for all $i$ and~$j$, then Alg.~\ref{alg:max-weight} is still a feasibility-optimal scheduling algorithm and Alg.~\ref{alg:approximate} is still a $\sqrt{M}$-approximate scheduling algorithm. 
\end{remark}

\section{Numerical results} \label{section:simulations}

\begin{figure}[t]
\centering
\includegraphics[width=.48\textwidth]{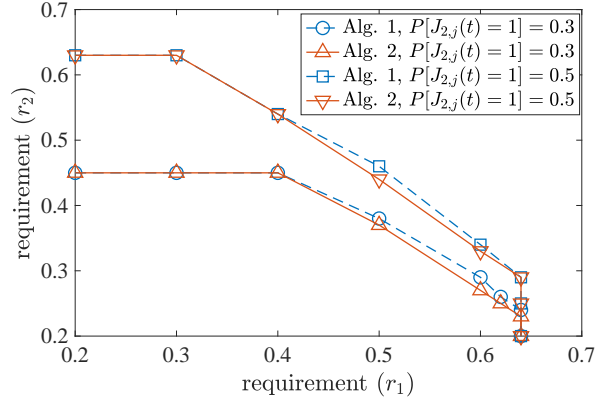}
\caption{Regions of achievable requirements by Algs.~\ref{alg:max-weight} and \ref{alg:approximate} for various task generation probabilities by application $A_2$, i.e., $P[J_{2,j}(t)=1]=0.3$ or $P[J_{2,j}(t)=1]=0.5$ for all $j$ and $t$.}
\label{fig:result12}	
\end{figure}
\begin{figure}[t]
	\centering
	\includegraphics[width=.48\textwidth]{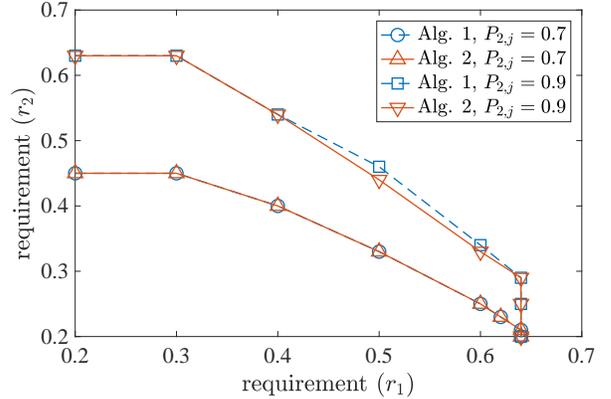}
	\caption{Regions of achievable requirements by Algs.~\ref{alg:max-weight} and \ref{alg:approximate} for various task completion probabilities by worker $W_2$, i.e., $P_{2,j}=0.7$ or $P_{2,j}=0.9$ for all $j$.}
	\label{fig:result45}	
\end{figure}
\begin{figure}[t]
	\centering
	\includegraphics[width=.48\textwidth]{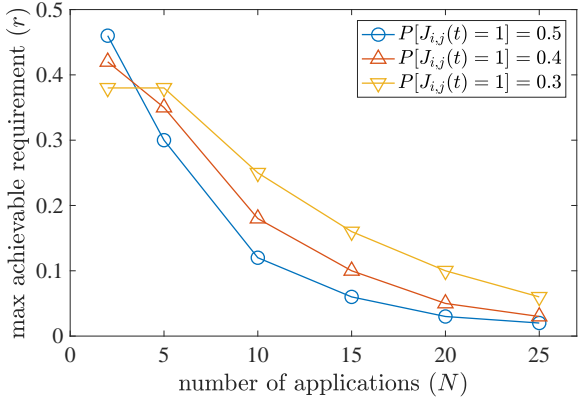}
	\caption{Maximum achievable requirement (for the case of $r_i=r$ for all $i$) versus  number of applications, for various task generation probabilities, i.e., $P[J_{i,j}(t)=1]=0.3$, $P[J_{i,j}(t)=1]=0.4$, or $P[J_{i,j}(t)=1]=0.5$, for all $i$, $j$, and~$t$.}
	\label{fig:result-large}	
\end{figure}

In this section, we investigate Algs.~\ref{alg:max-weight} and~\ref{alg:approximate} via computer simulations. First, we consider two applications and two workers. Fig.~\ref{fig:result12} displays the regions of achievable requirements  by both scheduling algorithms for various task generation probabilities by application $A_2$, when  $P[J_{1,j}(t)=1]=0.5$ and $P_{i,j}=0.9$ for all $i$, $j$, and $t$ are fixed. Fig.~\ref{fig:result45} displays the  regions of achievable requirements by both scheduling algorithms for various task completion probabilities by worker~$W_2$, when  $P[J_{i,j}(t)=1]=0.5$ and $P_{1,j}=0.9$ for all $i$, $j$, and~$t$ are fixed.  Each result $(r_1, r_2)$ marked in Figs.~\ref{fig:result12} or \ref{fig:result45} is the requirement such that the average number of completed jobs in 10,000 frames for application $A_1$ is at least  $r_1-0.01$ and that for application $A_2$ is  at least $r_2-0.01$. The both figures reflect that Alg.~\ref{alg:approximate} is not only computationally efficient but also can fulfill almost all requirements $\mathbf{r}$ within $\mathbf{R}_{\max}$  (achievable by Alg.~\ref{alg:max-weight}). 

Second, we consider more applications and more  workers with the same quantities, i.e.,  $N=M$. Moreover, all task completion probabilities are fixed to be 0.9, i.e., $P_{i,j}=0.9$ for all $i$ and $j$. Then, Fig.~\ref{fig:result-large} displays the maximum achievable requirements~$r$ (for the case of $r_i=r$ for all $i$) by Alg.~\ref{alg:approximate}, when all task generation probabilities are  the same.  In this case, an application generates a job in a frame with probability \mbox{$1-(1-P[J_{1,1}(1)=1])^N$}.  When $N=2$ in Fig.~\ref{fig:result-large}, the lower task generation probability the lower achievable requirement, because a lower task generation probability  generates fewer jobs.   
In contrast, when $N\geq5$ in Fig.~\ref{fig:result-large}, the lower task generation probability the higher achievable requirement, because fewer jobs  cause less interference. In other words, the interference becomes severe when $N\geq5$. Moreover, from Fig.~\ref{fig:result-large}, the maximum achievable requirement by Alg.~\ref{alg:approximate} appears to decrease super-linearly with the number of applications.


\section{Concluding remarks} \label{section:conclusion}
In this paper, we provided a framework for studying stochastic real-time jobs in unreliable workers with specialized functions. In particular, we developed two  algorithms for scheduling  real-time  jobs in  shared unreliable workers. While the proposed feasibility-optimal scheduling algorithm can support the largest region of applications' requirements, it has the notorious NP-hard issue. In contrast, the proposed approximate scheduling algorithm is not only simple, but also has a provable guarantee for the region of achievable requirements.  Moreover, we note that coding techniques have been exploited to alleviate  \textit{stragglers}  in distributed computing networks, e.g., \cite{lee2017speeding,yang2019timely}. Including coding design into our framework is promising.

\appendices

\section{Proof of Lemma~\ref{lemma:np}}\label{appendix:lemma:np}
We show a reduction from the set packing problem \cite{halldorsson1998independent}, where given a collection $\{\Psi_1, \cdots,  \Psi_n\}$ of \textit{non-empty} sets over a universal set $\{1, 2, \cdots, m\}$ for some positive integers $m$ and~$n$, the objective  is to identify a sub-collection  of disjoint sets in that collection such that the number of sets in the sub-collection  is maximized.

For the given instance of the set packing problem, we construct $n$ applications and $m$ workers in the  distributed computing network. Consider a fixed frame~$t$. In frame~$t$, application $A_i$ generates job $\Psi_i$. With the transformation, the set packing problem is equivalent to identifying a set of interference-free jobs in frame~$t$ such that number of jobs in that set is maximized.

Moreover, consider no job until frame~$t-1$,  identical requirements $r_i=r$ for all $i$, and identical task completion probabilities $P_{i,j}=1$ for all $i$ and $j$.  In this context, Eq.~(\ref{eq:opt}) in frame~$t$ becomes 
\begin{align}
\sum_{i: \mathbf{J}_i(t) \in \mathbf{D}(t)} r \cdot t,  \label{eq:opt2}
\end{align} 
because $Q_i(t)=r \cdot t$, $\mathbf{1}_{|\mathbf{J}_i(t)|\neq0}=1$ (due to non-empty sets $\Psi_i$ for all $i$), and  $ \prod_{j:J_{i,j}(t)=1}P_{i,j}=1$. As a result  of the constant  $r \cdot t$  in Eq.~(\ref{eq:opt2}), the objective of the combinatorial optimization problem in Alg.~\ref{alg:max-weight} in frame~$t$ becomes identifying a set of interference-free jobs such that the number of jobs in that set is maximized. 

Suppose there exists an algorithm such that the combinatorial optimization problem in Alg.~\ref{alg:max-weight} in frame~$t$ can be solved in polynomial time. Then,  the polynomial-time algorithm can identify a set $\mathbf{D}(t)$ for maximizing the value in Eq~(\ref{eq:opt2}); in turn, solves the set packing problem. That  contradicts to the NP-hardness of the set packing problem. 

Because the above argument is true for all frames~$t$, we conclude that the combinatorial optimization problem in Alg.~\ref{alg:max-weight} in frame~$t$ is NP-hard, for all~$t$.

\section{Proof of Lemma~\ref{lemma:apx}} \label{appendix:lemma:apx}
Consider a fixed queue size vector $\mathbf{Q}(t)$ in a fixed frame~$t$. Let $V_i(t)=Q_i(t) \prod_{j:J_{i,j}(t)=1}P_{i,j}$ for all $i=1, \cdots, N$. Without loss of generality, we can assume that $|\mathbf{J}_i(t)|\neq 0$ for all $i$ and further assume that $V_1(t)/\sqrt{|\mathbf{J}_1(t)|} \geq \cdots \geq V_N(t)/|\sqrt{|\mathbf{J}_N(t)|}$ (by reordering the job indices), i.e., Alg.~\ref{alg:approximate} processes job $\mathbf{J}_i(t)$  at the $i$-th iteration of Line~\ref{alg:approximate:for-begin}. Let $\mathbf{D}_2(t)$ be the decision of Alg.~\ref{alg:approximate} in frame~$t$ for the given queue size vector~$\mathbf{Q}(t)$. Then, we can express the value of Eq.~(\ref{eq:opt}) computed by Alg.~\ref{alg:approximate} as
\begin{align}
APX(t;\text{Alg.~\ref{alg:approximate}})=\sum_{i: \mathbf{J}_i(t) \in \mathbf{D}_2(t)} V_i(t). \label{eq:apx2}
\end{align}
 
Let $\mathbf{D}_1(t)$ be the decision of Alg.~\ref{alg:max-weight} in frame~$t$  for the given queue size vector $\mathbf{Q}(t)$. If  the conditions in Line~\ref{alg:approximate:condition} of Alg.~\ref{alg:approximate}  hold  for the  $i$-th iteration (i.e., $\mathbf{J}_i(t) \in \mathbf{D}_2(t)$), then we  let $\mathscr{C}_i=\{\mathbf{J}_k(t) \in \mathbf{D}_1(t): k\geq i, \mathbf{J}_{k}(t) \bigcap \mathbf{J}_i(t) \neq \emptyset\}$\footnote{Here, we use $\mathbf{J}_{k}(t) \bigcap \mathbf{J}_i(t)$ to represent the set of common workers for jobs $\mathbf{J}_k(t)$ and $\mathbf{J}_i(t)$.} be a set of jobs. The set $\mathscr{C}_i$ has the following properties: 
\begin{itemize}
	\item For  job $\mathbf{J}_k(t) \in \mathscr{C}_i$, we have
	\begin{align}
	\frac{V_k(t)}{\sqrt{|\mathbf{J}_k(t)|}} \leq \frac{V_i(t)}{\sqrt{|\mathbf{J}_i(t)|}},\label{eq:property1}
	\end{align}
	since $k \geq i$. 
	\item All jobs in $\mathscr{C}_i$ are interference-free,  i.e., they need different workers, since $\mathbf{J}_k(t) \in \mathbf{D}_1(t)$. Moreover, job $\mathbf{J}_k(t) \in \mathscr{C}_i$  needs at least one of the workers for $\mathbf{J}_i(t)$  (i.e., $\mathbf{J}_k(t) \bigcap \mathbf{J}_i(t) \neq \emptyset$). Thus, we have
	\begin{align}
	| \mathscr{C}_i|  \leq |\mathbf{J}_i(t)|. \label{eq:property2}
	\end{align}

	\item  	Since all jobs in $\mathscr{C}_i$ need different workers, and   there are $M$ workers, we have
	\begin{align}
	\sum_{k: \mathbf{J}_k(t) \in \mathscr{C}_i} |\mathbf{J}_k(t)| \leq M. \label{eq:property3}
	\end{align}

\end{itemize}

Note that $\mathbf{D}_1(t) \subseteq \bigcup_{i:\mathbf{J}_i(t) \in \mathbf{D}_2(t)} \mathscr{C}_i$. Thus, we can bound $OPT(t)$ computed by Alg.~\ref{alg:max-weight} by 
\begin{align}
OPT(t)\leq \sum_{i:\mathbf{J}_i(t) \in \mathbf{D}_2(t)} OPT_i, \label{eq:opt3}
\end{align}
where $OPT_i(t)=\sum_{k: \mathbf{J}_k(t) \in \mathscr{C}_i} V_k(t)$ for all $i$.

Furthermore, we can bound $OPT_i(t)$ for each $i$ by
\begin{align}
OPT_i(t) \mathop{\leq}^{(a)}&  \frac{V_i(t)}{\sqrt{|\mathbf{J}_i(t)|}}  \sum_{k: \mathbf{J}_k(t) \in \mathscr{C}_i}\sqrt{ |\mathbf{J}_k(t)|} \nonumber\\
\mathop{\leq}^{(b)}&\frac{V_i(t)}{\sqrt{|\mathbf{J}_i(t)|}} \sqrt{| \mathscr{C}_i|} \sqrt{\sum_{k: \mathbf{J}_k(t) \in \mathscr{C}_i} |\mathbf{J}_k(t)|}\nonumber\\
\mathop{\leq}^{(c)}& V_i(t) \sqrt{M}, \label{eq:ineq}
\end{align}
where (a) follows Eq.~(\ref{eq:property1}); (b) is due to the Cauchy-Schwarz inequality; (c) follows Eqs.~(\ref{eq:property2}) and (\ref{eq:property3}). 

Then, we can bound $OPT(t)$  by
\begin{align*}
OPT(t)\mathop{\leq}^{(a)}& \sum_{i:\mathbf{J}_i(t) \in \mathbf{D}_2(t)} OPT_i(t)\\
\mathop{\leq}^{(b)}& \sum_{i:\mathbf{J}_i(t) \in \mathbf{D}_2(t)} V_i(t) \sqrt{M} \\
\mathop{\leq}^{(c)}& \sqrt{M}\cdot APX(t;\text{Alg.~\ref{alg:approximate}}), 
\end{align*}
where (a) follows Eq.~(\ref{eq:opt3}); (b) follows Eq.~(\ref{eq:ineq}); (c) follows Eq.~(\ref{eq:apx2}). Because the above argument is true for all $\mathbf{Q}(t)$ and~$t$, the approximation ratio is $\sqrt{M}$.

\section{Proof of Theorem~\ref{theorem:approximation-ratio}}\label{appendix:theorem:approximation-ratio}
The proof of  Theorem~\ref{theorem:approximation-ratio}  needs the following technical lemma,  whose proof follows the line of \cite[Appendix 4.A]{neely2010stochastic} along with the i.i.d. property of  state $(\mathbf{J}_1(t), \cdots, \mathbf{J}_N(t))$ (as discussed in the proof of Theorem~\ref{theorem:optimal}) and the constant task completion probabilities $P_{i,j}$.
\begin{lemma} \label{lemma:stationary}
There exists a  \textit{stationary}  scheduling algorithm (i.e., decision $\mathbf{D}(t)$ depends on the state in frame~$t$ only) such that, for any  requirement~$\mathbf{r}$ interior of $\mathbf{R}_{\max}$,  all  queues in the virtual queueing network are stable, i.e., the stationary scheduling algorithm can fulfill the requirement~$\mathbf{r}$.
\end{lemma}

Moreover, we need the Lyapunov theory \cite[Thoereom 4.1]{neely2010stochastic} as stated in the following lemma, where we consider the Lyapunov function  $L(\mathbf{Q}(t))=\sum_{i=1}^N Q^2_i(t)$. 
\begin{lemma} \label{lemma:lyapunov}
Given scheduling algorithm~$\pi$ and requirement~$\mathbf{r}$, if there exist constants $B>0$ and $\epsilon>0$ such that
\begin{align*}
E[L(\mathbf{Q}(t+1))-L(\mathbf{Q}(t))|\mathbf{Q}(t)] \leq B-\epsilon\sum_{i=1}^N Q_i(t),
\end{align*} 		
for all frames~$t$, then all  queues in the virtual queueing network are stable, i.e., the scheduling algorithm~$\pi$ can fulfill the requirement $\mathbf{r}$. 
\end{lemma}

Then, we are ready to prove Theorem~\ref{theorem:approximation-ratio}. Suppose that requirement $\mathbf{r}=(r_1, \cdots, r_N)$ is interior of $\mathbf{R}_{\max}$. By  Lemma~\ref{lemma:stationary}, there exists a stationary scheduling algorithm that can fulfill the requirement~$\mathbf{r}$. We denote that stationary scheduling algorithm by~$\pi_{s}$. Moreover, since requirement $\mathbf{r}$ is interior of $\mathbf{R}_{\max}$, requirement $\mathbf{r}+\epsilon$ for some $\epsilon>0$ is also interior of $\mathbf{R}_{\max}$. By Lemma~\ref{lemma:stationary} again, the stationary scheduling algorithm~$\pi_{s}$ can fulfill requirement~$\mathbf{r}+\epsilon$, i.e., 
\begin{align}
\mathscr{N}_i(\pi_s) \geq r_i+\epsilon, \label{eq:stationary}
\end{align}
for all $i$.

Consider  requirement $\mathbf{r}'=(r'_1, \cdots, r'_N)$ where $r'_i=r_i/\sqrt{M}$ for all~$i$. 
Next, applying Lemma~\ref{lemma:lyapunov} to  Alg.~\ref{alg:approximate}, we  conclude that Alg.~\ref{alg:approximate} can fulfill requirement $\mathbf{r}'$ because
\begin{align*}
&E[L(\mathbf{Q}(t+1))-L(\mathbf{Q}(t))|\mathbf{Q}(t)]\\
\mathop{\leq}^{(a)}& B+2\sum_{i=1}^N Q_i(t) \cdot r'_i - 2\sum_{i=1}^N Q_i(t) \cdot E[e_i(t; \text{Alg.~\ref{alg:approximate}}) | \mathbf{Q}(t)]\\
\mathop{\leq}^{(b)}& B+2\sum_{i=1}^N Q_i(t)  \frac{r_i}{\sqrt{M}}- \frac{2}{\sqrt{M}}\sum_{i=1}^N Q_i(t)  E[e_i(t;\text{Alg.~\ref{alg:max-weight}}) | \mathbf{Q}(t)]\\
\mathop{\leq}^{(c)}& B+2\sum_{i=1}^N Q_i(t)  \frac{r_i}{\sqrt{M}} - \frac{2}{\sqrt{M}}\sum_{i=1}^N Q_i(t)  E[e_i(t;\pi_{s}) | \mathbf{Q}(t)]\\
\mathop{=}^{(d)}& B+2\sum_{i=1}^N Q_i(t)  \frac{r_i}{\sqrt{M}} - \frac{2}{\sqrt{M}}\sum_{i=1}^N Q_i(t)  \mathscr{N}_i(\pi_{s})\\
=&B + \frac{2}{\sqrt{M}}\sum_{i=1}^N Q_i(t)\left(r_i - \mathscr{N}_i(\pi_s)\right)\\
\mathop{\leq}^{(e)} & B - \frac{2\epsilon}{\sqrt{M}}\sum_{i=1}^N Q_i(t),
\end{align*}
where (a) follows \cite[Chapter 4]{neely2010stochastic} with some constant $B>0$; (b) is because $r'_i=r_i/\sqrt{M}$ and the approximation ratio of Alg.~\ref{alg:approximate} to Eq.~(\ref{eq:opt})  is $\sqrt{M}$ (as stated in Lemma~\ref{lemma:apx}); (c) is because Alg.~\ref{alg:max-weight} (in Line~\ref{alg:max-weight:d}) maximizes the value of $ \sum_{i=1}^N Q_i(t) \cdot E[e_i(t;\pi) | \mathbf{Q}(t)]$ among all possible scheduling algorithms~$\pi$; (d) is because decision $\mathbf{D}(t)$ under stationary scheduling algorithm~$\pi_s$ depends on the state only (regardless of the queue sizes) and also the state is i.i.d. over frames, yielding $E[e_i(t;\pi_{s}) | \mathbf{Q}(t)]=E[e_i(t;\pi_{s})]=\mathscr{N}_i(\pi_{s})$ for all $\mathbf{Q}(t)$ and~$t$; (e) follows Eq.~(\ref{eq:stationary}).

{\small
	\bibliographystyle{IEEEtran}
	\bibliography{IEEEabrv,ref}
}

\end{document}